\documentclass{llncs}

\title{A Foundation for Ledger Structures}
\author{Chad Nester\thanks{This research was supported by the ESF funded Estonian IT Academy research measure (project 2014-2020.4.05.19-0001).}}
\institute{Tallinn University of Technology, Estonia}

\usepackage{amsmath,amssymb}
\usepackage{mathpartir}
\usepackage{stmaryrd}
\usepackage{graphicx}
\usepackage{enumerate}
\usepackage{dsfont}
\usepackage[export]{adjustbox}
\usepackage[all]{xy}

\usepackage{float}
\floatstyle{boxed}
\restylefloat{figure}

\def \figPath {diagrams}

\newcommand{\X}{\mathbb{X}}

\newcommand{\alice}{\texttt{Alice}}

\newcommand{\bob}{\texttt{Bob}}

\newcommand{\carol}{\texttt{Carol}}

\begin{document}

\maketitle

\begin{abstract}
This paper introduces an approach to constructing ledger structures for cryptocurrency systems with basic category theory. Compositional theories of resource convertibility allow us to express the material history of virtual goods, and ownership is modelled by a free construction. Our notion of ownership admits an intuitive graphical representation through string diagrams for monoidal functors. 
\end{abstract}

\section{Introduction}

Modern cryptocurrency systems consist of two largely orthogonal parts: A consensus protocol, and the ledger structure it is used to maintain. While consensus protocols have received a lot of attention (see e.g. \cite{Nak08,Kia17}), the design space of the accompanying ledger structures is barely explored. The recent interest in smart contracts has led to the development of sophisticated ledger structures with complex behaviour (see e.g. \cite{Atz18,Woo14}). These efforts have been largely \emph{ad hoc}, and the resulting ledger structures are difficult to reason about. This difficulty also manifests in the larger system, which has contributed to several unfortunate incidents involving blockchain technology \cite{Atz17}.

A strong mathematical foundation for ledger structures would enable more rigorous development of sophisticated blockchain systems. Further, the ability to reason about the ledger at a high level of abstraction would facilitate analysis of system behaviour. This is important: users of the system must understand it in order to use it with confidence. The formalism we propose has an intuitive graphical representation, which would make this kind of rigorous operational understanding possible on a far wider scale that it would otherwise be. 

Blockchain systems are largely concerned with recording the material history of virtual objects, with a particular focus on changes in ownership. The resource theoretic interpretation of string diagrams for symmetric monoidal categories gives a precise mathematical meaning to this sort of material history. Building on this, we consider string diagrams augmented with extra information concerning the ownership of resources. We give these diagrams a precise mathematical meaning in terms of strong monoidal functors, drawing heavily on the work of \cite{McC11}, where our augmented diagrams originated. We show that an augmented resource theory has the same categorical structure as the original, in the sense that the two corresponding categories are equivalent. Finally, we give a simple example of a ledger structure using our machinery. 

\section{Monoidal Categories as Resource Theories}

We assume familiarity with some basic category theory, in particular with symmetric monoidal categories. A good reference is \cite{Mac71}. Throughout, we will write composition in diagrammatic order. That is, the composite of $f : X \to Y$ and $g : Y \to Z$ is written $fg : X \to Z$. We may also write $g \circ f : X \to Z$, but we will \emph{never} write $gf : X \to Z$. We will make heavy use of string diagrams for monoidal categories (see e.g. \cite{Sel10}), which we read from top to bottom (for composition) and left to right (for the monoidal tensor). Our string diagrams for ownership are in fact the string diagrams for monoidal functors of \cite{McC11}.

\subsection{Resource Theories}
We begin by observing (after \cite{Coe14}) that a symmetric strict monoidal category can be interpreted as a theory of resource convertibility: Each object corresponds to collection of resources with $A \otimes B$ denoting the collection composed of both $A$ and $B$ and the unit $I$ denoting the empty collection. Morphisms $f : A \to B$ are then understood as a way to convert the resources of $A$ to those of $B$.

For example, consider the free symmetric strict monoidal category on the set

\[\{\texttt{bread},\texttt{dough},\texttt{water},\texttt{flour},\texttt{oven}\}\]

\noindent of atomic objects, subject to the following additional axioms:

\begin{mathpar}
  \texttt{mix} : \texttt{water} \otimes \texttt{flour} \to \texttt{dough}

  \texttt{knead} : \texttt{dough} \to \texttt{dough}

  \texttt{bake} : \texttt{dough} \otimes \texttt{oven} \to \texttt{bread} \otimes \texttt{oven}
\end{mathpar}

This category can be understood as a theory of resource convertibility for baking bread. The morphism $\texttt{mix}$ represents the process of combining water and flour to form a bread dough, $\texttt{knead}$ the process of kneading the dough, and $\texttt{bake}$ the process of baking the dough in an oven to yield bread (and an oven).  While this model has many failings as a theory of bread, it suffices to illustrate the idea. The axioms of a symmetric strict monoidal category provide a natural scaffolding for this theory to live in. For example, consider the morphism

\[ (\texttt{bake} \otimes 1_{\texttt{dough}})(1_{\texttt{bread}} \otimes \sigma_{\texttt{oven},\texttt{dough}} \texttt{bake})\]

\noindent where $\sigma_{A,B} : A \otimes B \stackrel{\sim}{\longrightarrow} B \otimes A$ is the braiding. This morphism has type

\[\texttt{dough} \otimes \texttt{oven} \otimes \texttt{dough} \to \texttt{bread} \otimes \texttt{bread} \otimes \texttt{oven}\]

\noindent and describes the transformation of two pieces of dough into two loaves of bread by baking them one after the other in an oven. We obtain a string diagram for this morphism by drawing our objects as wires, and our morphisms as boxes with inputs and outputs. Composition is represented by connecting output wires to input wires, and we represent the tensor product of two morphisms by placing them beside one another. Finally, the braiding is represented by crossing the involved wires. For the morphism in question, we obtain:

\[ \includegraphics[height=4cm]{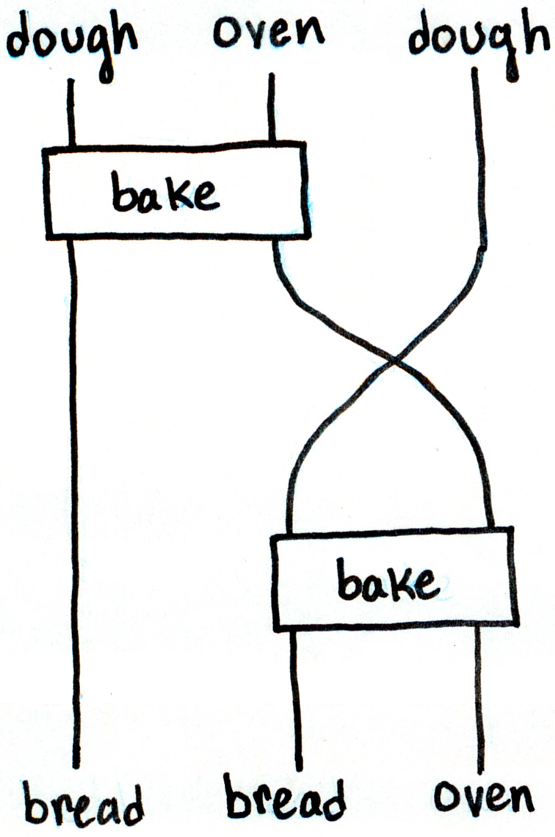} \]

We will think of our ledger systems in terms of such string diagrams: The state of the system is a string diagram describing the \emph{material history} of the resources involved, the available resources correspond to the output wires, and changes are effected by appending resource conversions to the bottom of the diagram. From now on we understand a \emph{resource theory} to be a symmetric strict monoidal category with an implicit resource-theoretic interpretation.

\subsection{How to Read Equality} 

Suppose we have a resource theory $\X$, and two resource transformations $f,g : A \to B$. Each of $f$ and $g$ expresses a different way to transform an instance of resource $A$ into an instance of resource $B$, but these may not have the same effect. For example, consider $\texttt{knead} : \texttt{dough} \to \texttt{dough}$ and $1_{\texttt{dough}} : \texttt{dough} \to \texttt{dough}$ from our resource theory of bread. Clearly these should not have the same effect on the input dough. This is reflected in our resource theory in the sense that they are not made equal by its axioms. For contrast, we can imagine a (somewhat) reasonable model of baking bread in which there is no difference between kneading the dough once and kneading it many times. We could capture this in our resource theory of baking bread by imposing the equation
\[ \texttt{knead} = \texttt{knead} \circ \texttt{knead}\]
In this new resource theory, our equation tells us that kneading dough once has the same effect as kneading it twice, or three times, and so on, since the corresponding morphisms of the resource theory are made equal by its axioms. Of course, the material history described by $\texttt{knead} \circ \texttt{knead}$ is not identical to that described by $\texttt{knead}$. In the former case, the kneading process has been carried out twice in sequence, while in the latter case it has only been carried out once. That these morphisms are equal merely means that the effect of each sequence of events on the dough involved is the same. 

We adopt the following general principle in our design and understanding of resource theories: \textit{Two transformations should be equal precisely when they have the same effect on the resources involved}.

We further illustrate this by observing that, by the axioms of a symmetric monoidal category (specifically, by naturality of braiding), the following two transformations in the resource theory of baking (expressed as string diagrams) are equal. The transformation on the left describes baking two loaves of bread by first mixing and kneading two batches of dough before baking them in sequence, while the transformation on the right describes baking two loaves of bread by mixing, kneading, and baking the first batch of dough, and \emph{then} mixing, kneading, and baking the second batch. Thus, according to our resource theory the two procedures will yield the same result -- not an entirely unreasonable conclusion! 

\[\includegraphics[height=10cm]{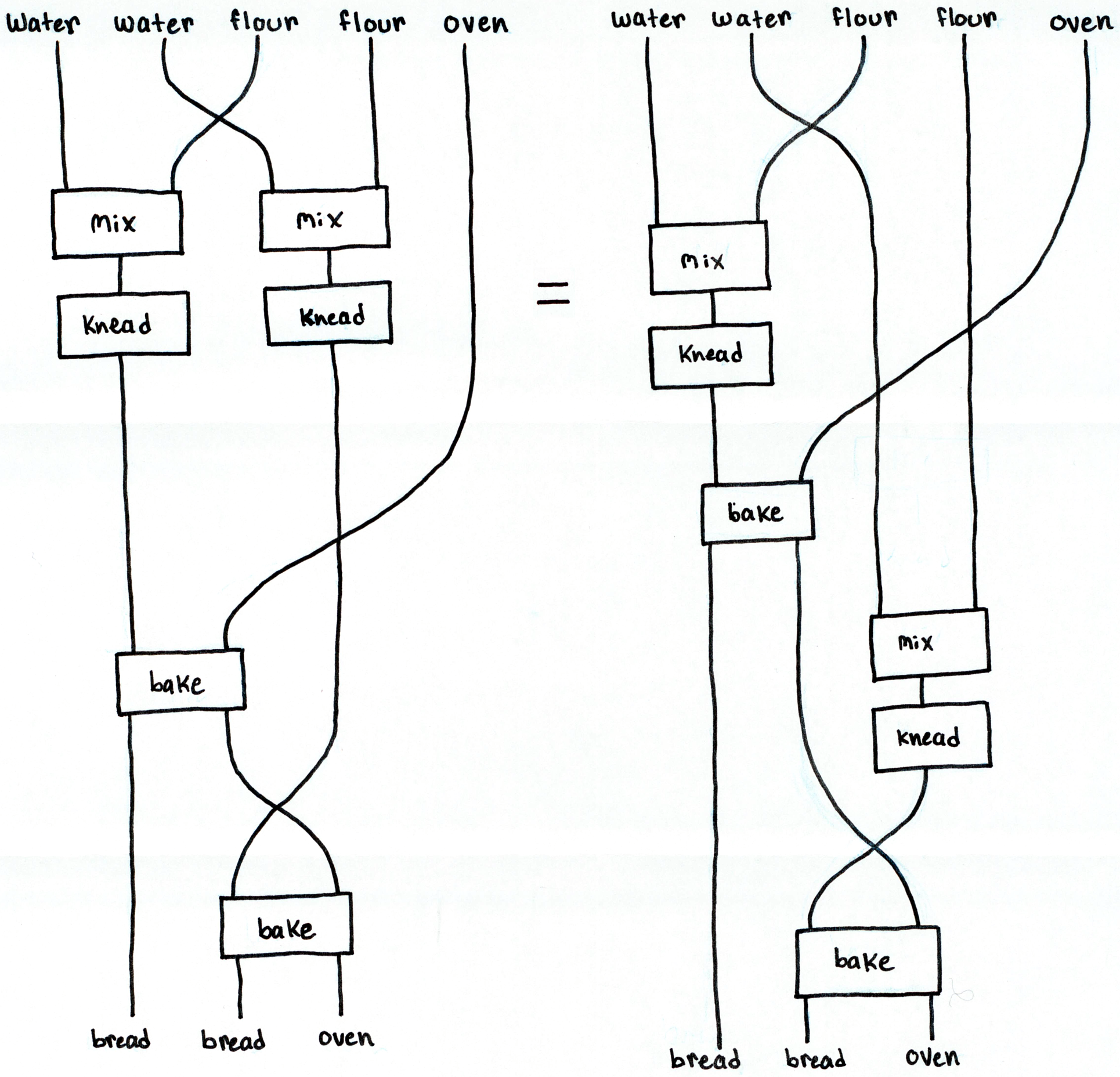}\]

\section{String Diagrams for Ownership}

Ledgers used by blockchain systems are largely concerned with \emph{ownership}. For example, in the Bitcoin system, each coin is associated with a computable function called the \emph{validator}, which is used to control access to it. Anyone who wishes to use the coin must supply input data, called a \emph{redeemer}, and the system only allows them to use the coin in question in case running the validator on the redeemer terminates in a fixed amount of time. If the validator is defined only on the data that results from $\alice$ digitally signing a nonce generated by the system, then that coin can only be used by $\alice$, who then effectively owns it. 

Different use cases call for different authentication schemes. For example, a proposed application of blockchain technology is to improve supply chain accountability by requiring participants to log any transfers and transformations of material on a public ledger (see e.g. \cite{Jab18,Sta17}). Here ownership implies responsibility, and so for Alice to log the transfer of, say, a ton of steel to Bob, \emph{both} Alice and Bob must ratify the transfer via digital signature.

What different use cases have in common is that the resources of the ledger system are associated with ownership data. We leave the interpretation of this ownership data, including the specific details of the authentication scheme unspecified, instead giving a structural account of resource ownership. We develop our account of resource ownership intuitively, and somewhat informally, by introducing addtional features to string diagrams. This is made fully formal in the next section. 

\subsection{Ownership and Collection Management}

Begin by assuming a theory of resources $\X$, and a collection $\mathcal{C}$ of potential resource owners, each of which we associate with a colour for use in our diagrams. Suppose for the remainder that $\alice$, $\bob$, and $\carol$ range over $\mathcal{C}$, and are associated with colours as follows: 

\[ \includegraphics[height=0.8cm]{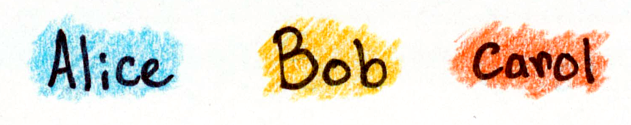} \]

Our goal will be to construct a new theory of resources in which resources and transformations are associated with (owned and carried out by) elements of $\mathcal{C}$. The objects of our new resource theory will be collections of owned objects of $\X$. That is, for each object $X$ of $\X$ and each $\alice \in \mathcal{C}$ we have an object $X^{\alice}$, which we interpret as an instance of resource $X$ owned by $\alice$, along with the empty collection $I$ and composite collections $X^{\alice} \otimes Y^{\bob}$, in which $\alice$'s instance of $X$ exists alongside an instance of $Y$ owned by $\bob$. 

Similarly, for each transformation $f : X \to Y$ in $\X$, we ask for transformations $f^{\alice} : X^{\alice} \to Y^{\alice}$ and $f^{\bob} : X^{\bob} \to Y^{\bob}$ for all $\alice,\bob \in \mathcal{C}$, whose presence we interpret as the ability of each owner to effect all possible transformations of resources they own. We draw these annotated transformations as, respectively:

\[ \includegraphics[height=2cm]{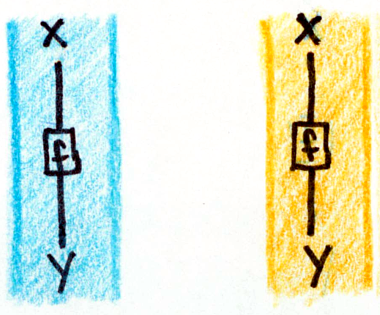} \]

Since we are building a theory of resources we must end up with a symmetric monoidal category, so we also assume the presence of the associated morphisms, such as $f^{\alice} \otimes g^{\bob}$ and $\sigma_{X^{\alice},Y^{\bob}}$. 




Next, we account for the formal difference between $X^{\alice} \otimes Y^{\alice}$ and $(X \otimes Y)^{\alice}$. In both situations $\alice$ owns an $X$ and a $Y$, but in the former they are formally grouped together, while in the latter they are formally separated. We understand this formal grouping of $\alice$'s assets by analogy with physical currency. The situation in which $\alice$'s assets are separated is like $\alice$ having two coins worth one euro, while the situation in which they are grouped together is like $\alice$ having one coin worth two euros. In both cases, $\alice$ posesses two euros, but the difference is important: $\alice$ cannot give $\bob$ half of the two euro coin, but can easily give $\bob$ one of the two one euro coins. This distinction is also present in cryptocurrency systems, where there is an operational difference between having funds spread across many addresses and having them collected at one address. Reflecting both the reality of such systems and the principle that one ought to be able to freely reconfigure the formal grouping of things that they own, we ask that for each $X,Y$ objects of $\X$ and each $\alice \in \mathcal{C}$ our new resource theory has morphisms $\phi_{X,Y} : X^{\alice} \otimes Y^{\alice} \to (X \otimes Y)^{\alice}$ and $\psi_{X,Y} : (X \otimes Y)^{\alice} \to X^{\alice} \otimes Y^{\alice}$. We draw these morphisms, respectively, as follows:

\[ \includegraphics[height=2cm]{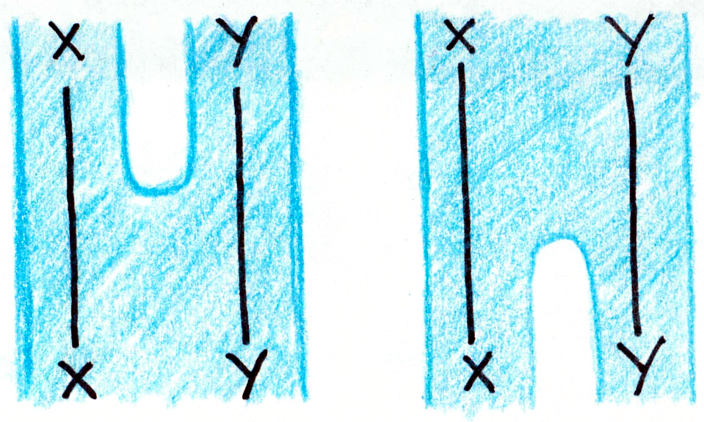} \]

These changes of formal grouping should not interact with the resource transformations of our original theory $\X$, since it ought not matter whether $\alice$ combines (splits) her resources before or after transforming them. That is, we we require:
\begin{enumerate}[]
\item \textbf{[G.1]} $\phi^{\alice}_{X,Y}(f \otimes g)^{\alice} = (f^{\alice} \otimes g^{\alice})\phi^{\alice}_{X',Y'}$
\item \textbf{[G.2]} $(f \otimes g)^{\alice}\psi^{\alice}_{X',Y'} = \psi^{\alice}_{X,Y}(f^{\alice} \otimes g^{\alice})$
\end{enumerate}
\[ \includegraphics[height=2cm]{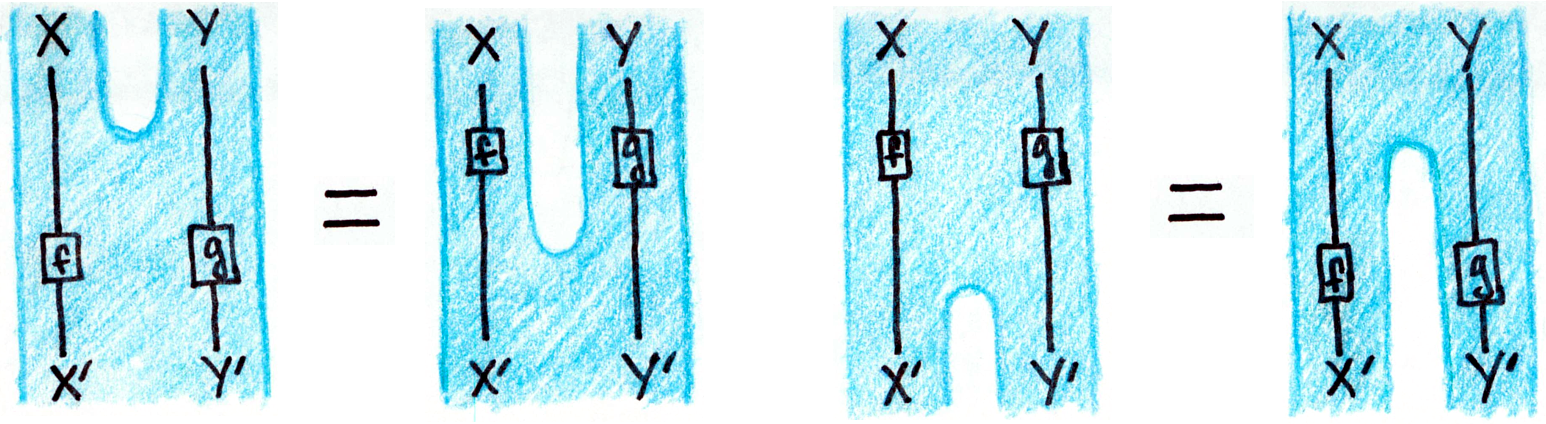}\]

As it stands, there are many non-equal ways for $\alice$ to reconfigure the formal grouping of their assets. Since these should all have the same effect, we need them all to be equal as morphisms in our resource theory. It suffices to ask that the $\phi^{\alice}$ and $\psi^{\alice}$ maps give, respectively, associative and coassociative operations, and that they are mutually inverse. That is (associativity and coassociativity):
\begin{enumerate}[]
\item \textbf{[G.3]} $(\phi^{\alice}_{X,Y} \otimes 1_Z^{\alice})\phi^{\alice}_{X\otimes Y,Z} = (1_X^{\alice} \otimes \phi^{\alice}_{Y,Z})\phi^{\alice}_{X,Y\otimes Z}$
\item \textbf{[G.4]} $\psi_{X\otimes Y,Z}^{\alice}(\psi_{X,Y}^{\alice} \otimes 1_Z^{\alice}) = \psi_{X,Y\otimes Z}^{\alice}(1_X^{\alice} \otimes \psi_{Y,Z}^{\alice})$
\end{enumerate}
\[ \includegraphics[height=2cm]{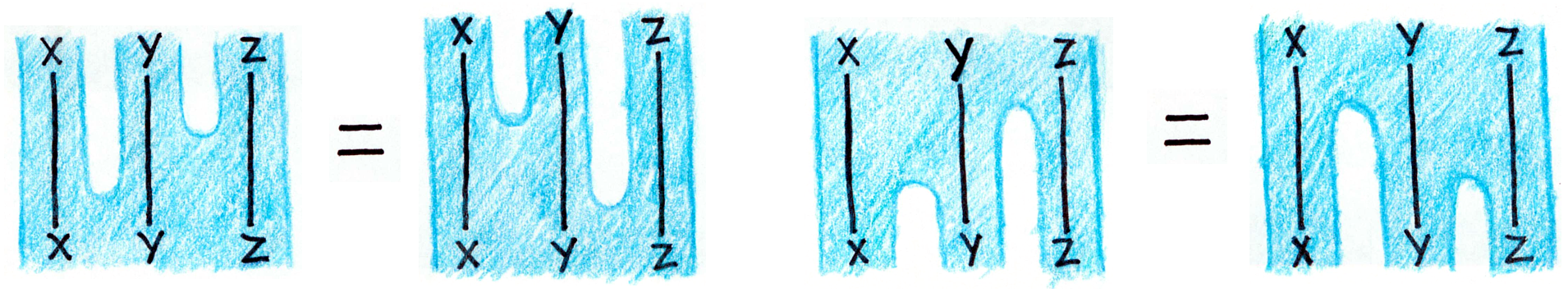} \]
and (mutually inverse): 
\begin{enumerate}[]
\item \textbf{[G.5]} $\psi^{\alice}_{X,Y}\phi^{\alice}_{X,Y} = 1_{X \otimes Y}^{\alice}$
\item \textbf{[G.6]} $\phi^{\alice}_{X,Y}\psi^{\alice}_{X,Y} = 1_X^{\alice} \otimes 1_Y^{\alice}$
\end{enumerate}
\[ \includegraphics[height=2cm]{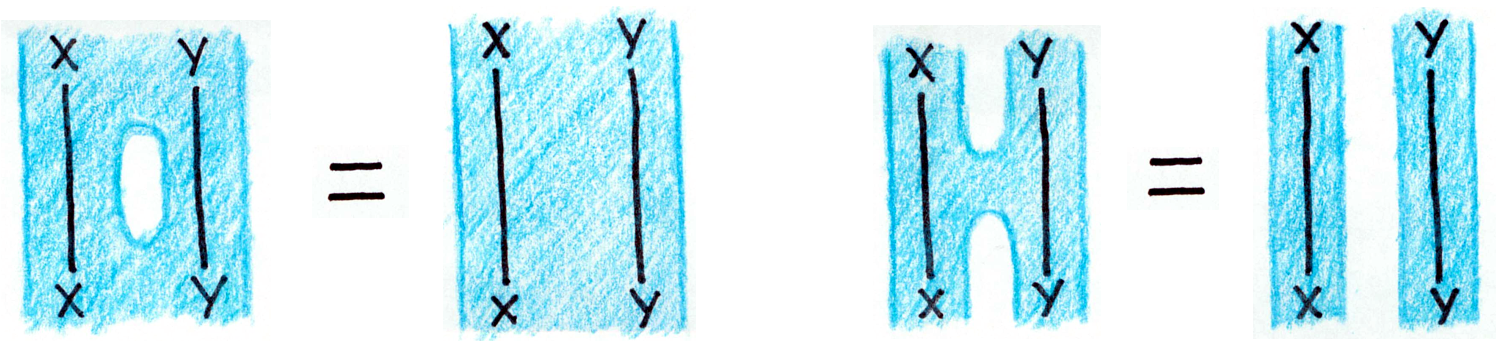} \]

To complete our treatment of these formal resource groupings, we must deal with the empty case $I^{\alice}$. We insist that $\alice$ may freely create and destroy such empty collections via morphisms $\phi^{\alice}_I : I \to I^{\alice}$ and $\psi^{\alice}_I : I^{\alice} \to I$:

\[ \includegraphics[height=2cm]{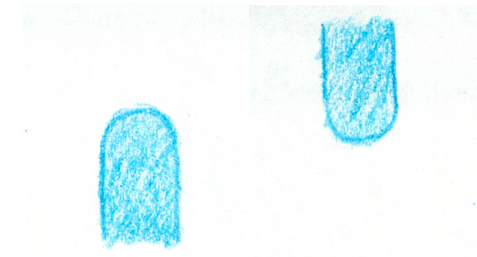} \]

\noindent subject to the following axioms, which state that adding or removing nothing from a group or resources has the same effect as doing nothing, and that $\phi_I$ and $\psi_I$ are mutually inverse, which together ensure that even with $\phi_I$ and $\psi_I$ in the mix, any two formal regroupings with the same domain and codomain are equal. 
\begin{enumerate}[]
\item \textbf{[G.7]} $(\phi_I^{\alice} \otimes 1_X^{\alice})\phi_{I,X}^{\alice} = 1_X^{\alice} = (1_X^{\alice} \otimes \phi_I^{\alice})\phi_{X,I}^{\alice}$
\item \textbf{[G.8]} $\psi_{I,X}^{\alice}(\psi_I^{\alice} \otimes 1_X^{\alice}) = 1_X^{\alice} = \psi_{X,I}^{\alice}(1_X^{\alice} \otimes \psi_I^{\alice})$
\item \textbf{[G.9]} $\phi_I^{\alice}\psi_I^{\alice} = 1_I$
\item \textbf{[G.10]} $\psi_I^{\alice}\phi_I^{\alice} = 1_I^{\alice}$
\end{enumerate}
\[ \includegraphics[height=4cm]{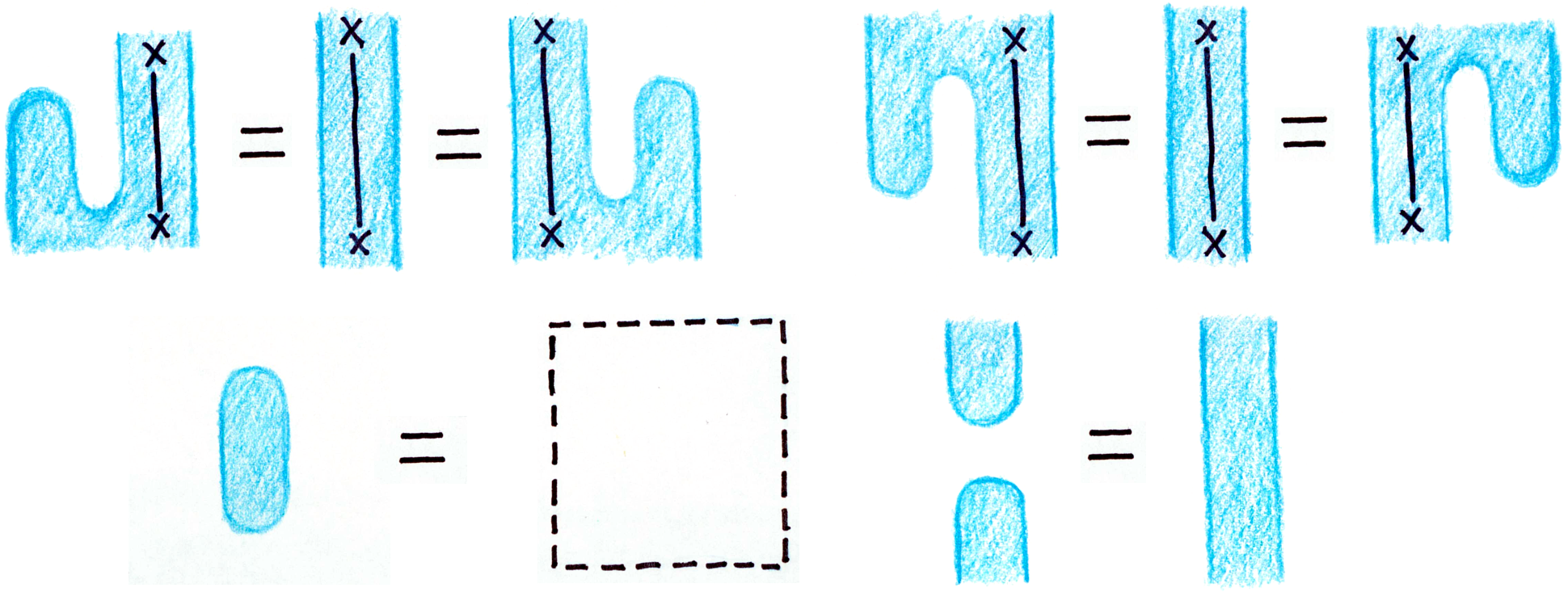} \]

Finally, we ask that $\phi$ and $\psi$ behave coherently with respect to the symmetry maps. It suffices to require that
\begin{enumerate}[]
\item \textbf{[G.11]} $\phi_{X,Y}^{\alice}\sigma_{X,Y}^{\alice} = \sigma_{X^{\alice},Y^{\alice}}\phi_{Y,X}^{\alice}$
\end{enumerate}
\[ \includegraphics[height=2.3cm]{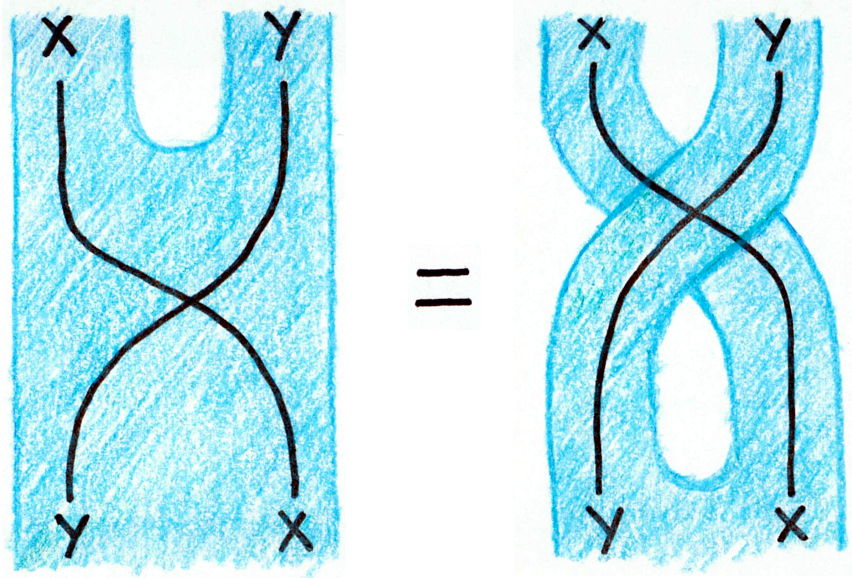} \]

\subsection{Change of Ownership}

Of course, ownership is not static over time. We require the ability the \emph{change} the owner of a given collection of resources. To this end we add morphisms $\gamma_X^{\alice,\bob} : X^{\alice} \to X^{\bob}$ to our new resource theory for each object $X$ of $\X$, each $\alice,\bob \in \mathcal{C}$. We depict these new morphisms in our string diagrams as follows:

\[ \includegraphics[height=2cm]{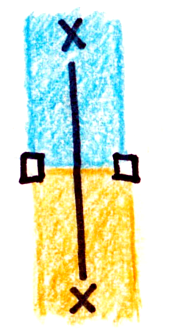} \]

As with regrouping, change of ownership should not interact with resource transformations, in the sense that:
\begin{enumerate}[]
\item \textbf{[O.1]} $f^{\alice}\gamma_Y^{\alice,\bob} = \gamma_X^{\alice,\bob}f^{\bob}$
\end{enumerate}
\[ \includegraphics[height=2cm]{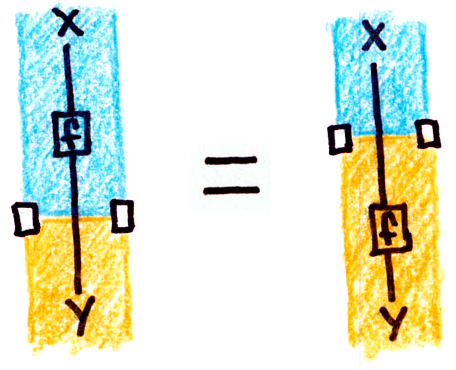} \]

Further, change of ownership must behave coherently with respect to the regrouping morphisms in the sense that:
\begin{enumerate}[]
\item \textbf{[O.2]} $\phi_{X,Y}^{\alice}\gamma_{X\otimes Y}^{\alice,\bob} = (\gamma_X^{\alice,\bob}\otimes \gamma_Y^{\alice,\bob})\phi_{X,Y}^{\bob}$
\item \textbf{[O.3]} $\gamma_{X \otimes Y}^{\alice,\bob}\psi^{\bob}_{X,Y} = \psi^{\alice}_{X,Y}(\gamma^{\alice,\bob}_X \otimes \gamma^{\alice,\bob}_Y)$
\end{enumerate}
\[ \includegraphics[height=2cm]{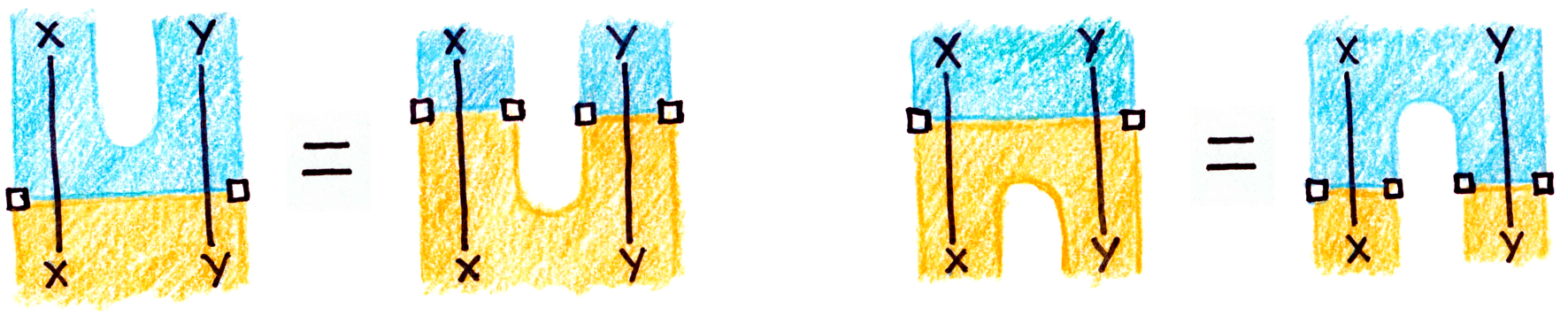} \]

For completeness, we axiomatize the interaction of change of ownership with empty collections by requiring that:
\begin{enumerate}[]
\item \textbf{[O.4]} $\phi_I^{\alice}\gamma_I^{\alice,\bob} = \phi_I^{\bob}$
\item \textbf{[O.5]} $\gamma_I^{\alice,\bob}\psi_I^{\bob} = \psi_I^{\alice}$
\end{enumerate}
\[ \includegraphics[height=2cm]{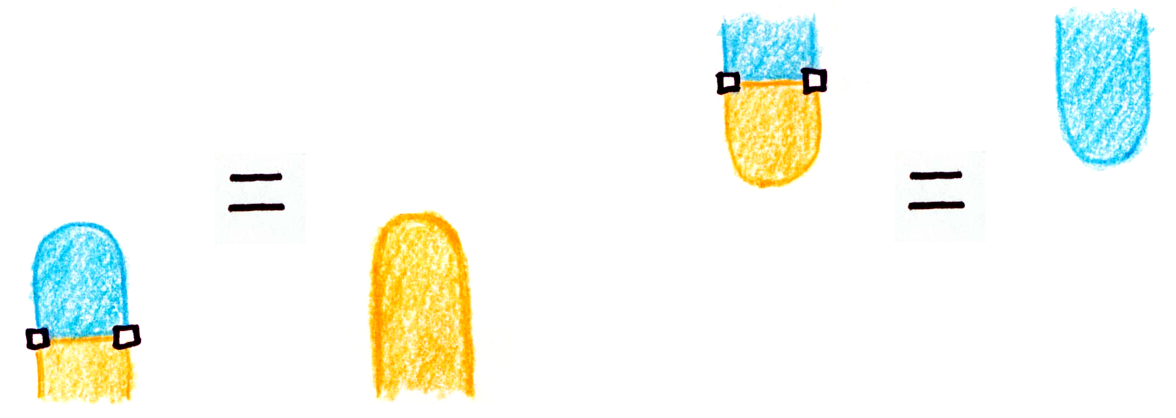} \]

Finally, we insist that if $\alice$ gives something to $\bob$, and $\bob$ then gives it to $\carol$, this has the same effect as $\alice$ giving the thing directly to $\carol$. Similarly, if $\alice$ gives something to $\alice$, we insist that this has no effect. 
\begin{enumerate}[]
\item \textbf{[O.6]} $\gamma^{\alice,\bob}_X\gamma^{\bob,\carol}_X = \gamma_X^{\alice,\carol}$
\item \textbf{[O.7]} $\gamma^{\alice,\alice}_X = 1_X^{\alice}$
\end{enumerate}

\[ \includegraphics[height=2cm]{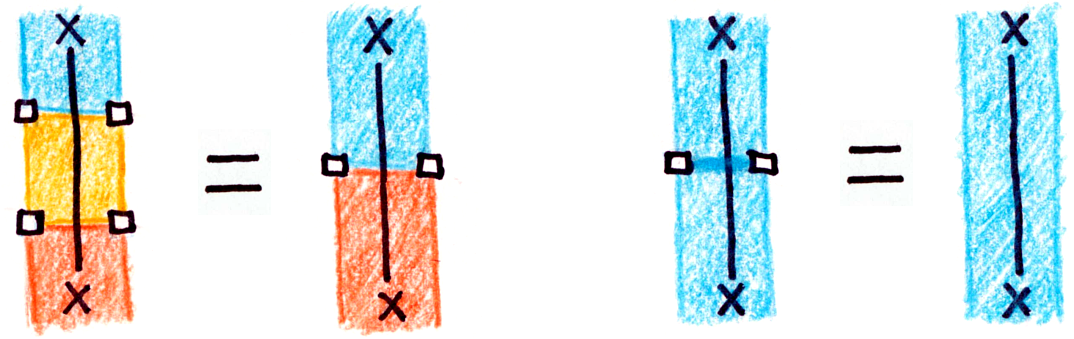} \]

We end up with a rather expressive diagrammatic language. For example, if we begin with the resource theory of bread, then our new resource theory is powerful enough to show:

\[ \includegraphics[height=8cm]{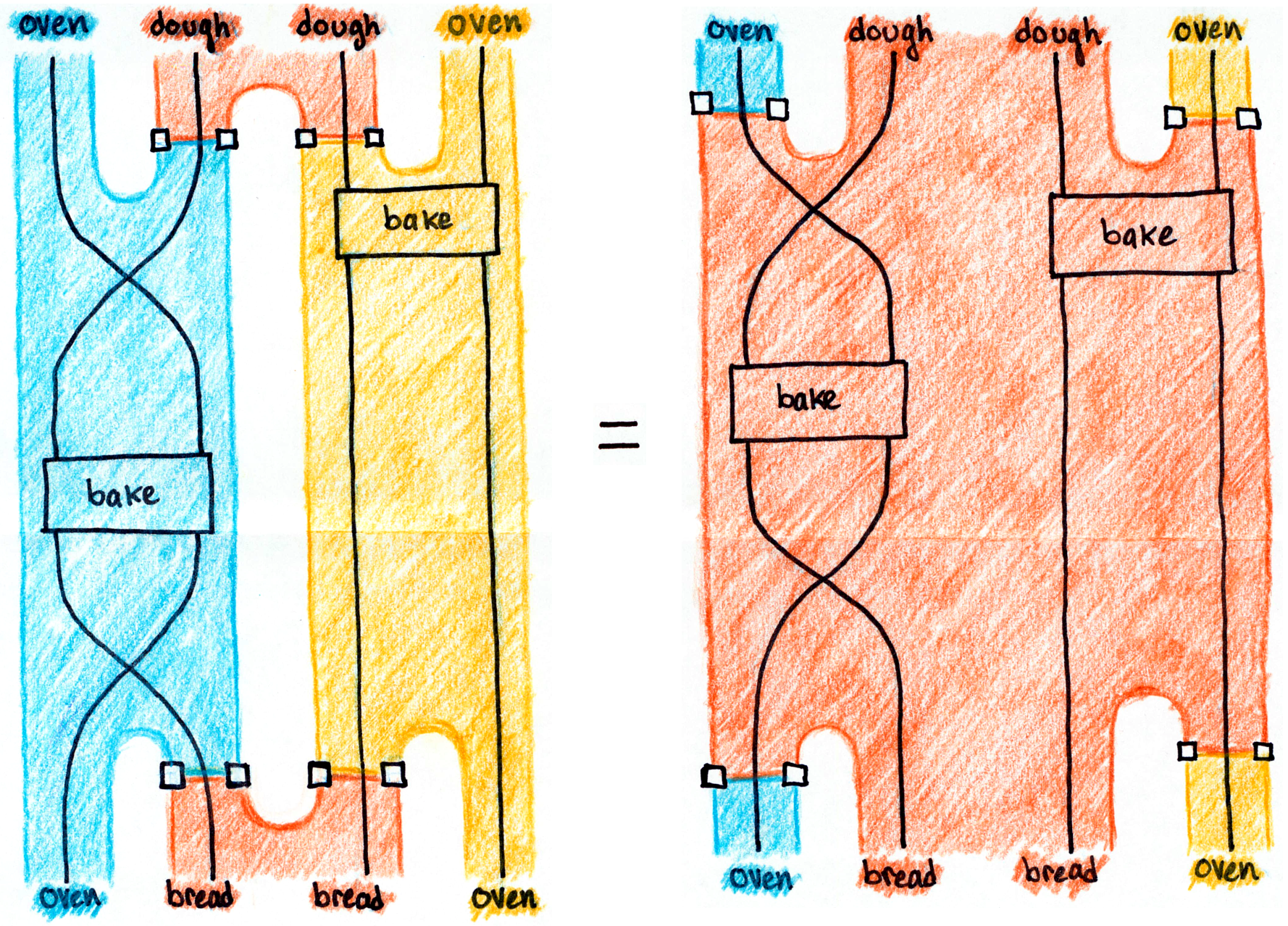} \]

which captures the fact that the sequence of events on the left in which $\carol$ gives $\alice$ and $\bob$ each a portion of dough to bake in their ovens, after which they give the resulting bread to $\carol$ \emph{has the same effect} as the sequence of events on the right in which $\alice$ and $\bob$ give their ovens to $\carol$, who bakes the portions of dough herself before returning the ovens to their original owners. Notice that our diagrammatic representation of this is \emph{much} easier to understand than the corresponding terms in linear syntax!

\section{Categorical Semantics}

In this section we show how our augmented string diagrams can be given precise mathematical meaning. Specifically, from a resource theory and a set whose elements we think of as entities capable of owning resources, we construct a new resource theory in which all resources are owned by some entity. We finish by showing how to model a simple cyrptocurrency ledger with our machinery. 

\subsection{Interpreting String Diagrams with Ownership}

If $\X$ is a theory of resources and $\mathcal{C}$ is our set, we treat $\mathcal{C}$ as the corresponding discrete category, writing $A : A \to A$ for the identity maps, and form the product category $\X \times \mathcal{C}$. Write objects and maps of this product category as $X^A = (X,A)$ and $f^A = (f,A)$ respectively. Now, define $\mathcal{C}(\X)$ to be the free strict symmetric monoidal category on $\X \times \mathcal{C}$ subject to the following additional axioms:
\begin{mathpar}
  \inferrule{A \in \mathcal{C} \\ X,Y \text{ objects of } \X}{\phi^A_{X,Y} : X^A \otimes Y^A\to (X \otimes Y)^A \text{ in } \mathcal{C}(\X)}

  \inferrule{A \in \mathcal{C}}{\phi_I^A : I \to I^A \text{ in } \mathcal{C}(\X)}

  \inferrule{A \in \mathcal{C} \\ X,Y \text{ objects of } \X}{\psi^A_{X,Y} : (X \otimes Y)^A \to X^A \otimes Y^A \text{ in } \mathcal{C}(\X)}

  \inferrule{A \in \mathcal{C}}{\psi_I^A : I^A \to I \text{ in } \mathcal{C}(\X)}

  \inferrule{A,B \in \mathcal{C} \\ X \text{ an object of } \X}{\gamma_X^{A,B} : X^A \to X^B \text{ in } \mathcal{C}(\X)}
\end{mathpar}
and subject to equations \textbf{[G.1--11]} and \textbf{[O.1--7]} for $\alice,\bob,\carol \in \mathcal{C}$, $X,Y,Z$ objects of $\X$, and $f,g$ morphisms of $\X$. 

Clearly, $\mathcal{C}(\X)$ is the new resource theory our coloured string diagrams live in. We think of objects $X^A$ and morphisms $f^A$ as being owned and carried out, respectively, by $A \in \mathcal{C}$. The free monoidal structue gives us the ability to compose such transformations sequentially and in parallel, and the additional axioms ensure our ownership interpretation of $\mathcal{C}(\X)$ is reasonable.

We can characterize the category-theoretic effect of axioms \textbf{[G.1--11]} and \textbf{[O.1--5]} as follows:
  
\begin{proposition}
  For any symmetric monoidal category $\X$ and any set $\mathcal{C}$, there is a strong symmetric monoidal functor
  \[ A : \X \to \mathcal{C}(\X) \]
  for each $A \in \mathcal{C}$. Further, there is a monoidal and comonoidal natural transformation
  \[\gamma^{A,B} : A \to B\]
  between the functors corresponding to any two $A,B \in \mathcal{C}$.
\end{proposition}
\begin{proof}
Define $A : \X \to \mathcal{C}(\X)$ by $A(X) = (X,A)$ on objects, and $A(f) = (f,A)$ on maps. For identity maps, $A(1_X) = (1_X,A) = 1_{(X,A)} = 1_{A(X)}$ since $(1_X,A)$ is the identity on $(X,A)$ in $\X \times \mathcal{C}$. For composition, $A(fg) = (fg,A) = (f,A)(g,A) = A(f)A(g)$.Thus $A$ defines a functor. $A$ is strong symmetric monoidal via the $\phi^A$ and $\psi^A$ maps together with \textbf{[G.1]} through \textbf{[G.11]}. Consider $A,B : \X \to \mathcal{C}(\X)$ corresponding to $A,B \in \mathcal{C}$. Define $\gamma^{A,B} : A \to B$ to have components $\gamma^{A,B}_X$. Then $\gamma^{A,B}$ is a monoidal and comonoidal via \textbf{[O.1]} through \textbf{[O.5]}. \qed
\end{proof}      

Notice that we did not use \textbf{[O.6--7]} above. These axioms are motivated by our desire to model resource ownership, but they have an important, if subtle, effect on the theory: they allow us to show that $\X$ and $\mathcal{C}(\X)$ are equivalent as categories. This means that any suitably categorical structure is present in $\X$ if and only if it is present in $\mathcal{C}(\X)$ as well. For example, products in $\X$ manifest as products in $\mathcal{C}(\X)$, morphisms that are monic in $\X$ remain monic in $\mathcal{C}(\X)$, and so on. We may be confident that our addition of ownership information has not broken any of the structure of $\X$, or added anything superfluous!

\begin{proposition}
There is an adjoint equivalence between $\X$ and $\mathcal{C}(\X)$ for each functor corresponding to some $A \in \mathcal{C}$.
\end{proposition}  
\begin{proof}
   
  We show that each $A : \X \to \mathcal{C}(\X)$ is fully faithful, and essentially surjective, beginning with the latter. To that end, suppose that $P$ is an object of $\mathcal{C}(\X)$. We proceed by structural induction: If $P$ is $I$, then $\phi_0$ witnesses $I \simeq A(I)$. If $P$ is an atom $(X,A)$, then $(X,A) = A(X)$. If $P$ is $Q \otimes R$ for some $Q,R$, then by induction we have that $Q \simeq A(X_1)$ and $R \simeq A(X_2)$ for some objects $X_1,X_2$ of $\X$. We may now form

\centerline{\xymatrix{
      Q \otimes R \simeq A(X_1) \otimes A(X_2) \ar[r]^(0.66){\phi^A_{X_1,X_2}} & A(X_1 \otimes X_2)
  }}

\noindent which witnesses $P \simeq A(X_1 \otimes X_2)$. Thus, $A$ is essentially surjective. To see that $A$ is fully faithful, let $U : \mathcal{C}(\X) \to \X$ be the obvious forgetful functor. The required bijection $\X(X,Y) \simeq \mathcal{C}(\X)(A(X),A(Y))$ is given by $A$ in one direction and $U$ in the other. It sufffices to show that any morphism $h : A(X) \to A(Y)$ with $U(h) = f$ is such that $h = A(f)$. Notice that since each $\gamma^{A,B}$ is a monoidal and comonoidal natural transformation, there is a term equal to $h$ in which all $\gamma$ morphisms occur before all other morphisms (in the sense that $f$ occurs before $g$ in $fg$). Since $h : A(X) \to A(Y)$ we know that in this equal term the composite of the $\gamma$ must have type $A(X) \to A(X)$, and must therefore be the identity by repeated application of \textbf{[O.6]} and \textbf{[O.7]}. This gives a term $h'$ containing no $\gamma$ maps with $h' = h$. Similarly, since the various $\phi$ and $\psi$ morphisms are natural transformations, we may construct a term $h''$ by collecting all instances of $\phi$ and $\psi$ terms at the beginning of $h'$. Once collected there, the composite of all the $\phi$ and $\psi$ must have type $A(X) \to A(X)$, and is therefore equal to the identity. At this point we know that $h'' : A(X) \to A(Y)$ is such that $h'' = A(f_1)\cdots A(f_n)$ for some $f_1,\ldots,f_n$ in $\X$. By assumption $f = U(h) = U(h'') = f_1\cdots f_n$, and therefore $h'' = A(f)$. \qed
\end{proof}

\subsection{A Simple Example}


In this section we attempt to demonstrate the relevance of the above techniques to the cryptocurrency world by building a resource theory that models a simple ledger structure along the lines of Bitcoin \cite{Nak08}. Let $\mathds{1}$ be the trivial category, with one object, $1$, and one morphism, the identity $1_1$. Define $\mathbb{N}$ to be the free symmetric strict monoidal category on $\mathds{1}$, write $0$ for the monoidal unit of $\mathbb{N}$, and $n$ for the $n$-fold tensor product of $1$ with itself for all natural numbers $n \geq 1$. Notice that $n + m$ is $n \otimes m$. We will think of the objects $n$ of $\mathbb{N}$ where $n \geq 1$ as \emph{coins}. Of course, $0 = I$ represents the situation in which no coin in present. 

Define $\mathbb{N}_\nu$ to be the result of formally adding a morphism $\nu : 0 \to 1$ to $\mathbb{N}$, write $\nu_0 = 1_0 : 0 \to 0$, and $\nu_n : 0 \to n$ for the $n$-fold tensor product of $\nu$ with itself for $n \geq 1$. These morphisms confer the ability to create new coins, so we imagine their use would be restricted in practice. We will not ask for the ability to destroy coins, although there would be no theoretical obstacle to doing so. 

Now, let $\mathcal{C}$ be a collection of colours, which we can think of as standing in for cryptographic key pairs, or simply entities capable of owning coins. Consider $\mathcal{C}(\mathbb{N}_\nu)$. Objects are lists $n_1^{c_1} \otimes \cdots \otimes n_k^{c_k}$, which we interpret as lists of coins, where $n_i^{c_i}$ is a coin of value $n_i$ belonging to $c_i \in \mathcal{C}$. The morphisms are either $\nu_n^c$ for some $c \in \mathcal{C}$, the structural morphisms of a monoidal category, or the $\phi,\psi$, and $\gamma$ morphisms added by our construction. For $n,m \in \mathbb{N}$ and $\alice,\bob \in \mathcal{C}$, the maps $\phi_{n,m}^{\alice} : n^{\alice} \otimes m^{\alice} \to (n + m)^{\alice}$ and $\psi^{\alice}_{n,m} : (n + m)^{\alice} \to n^{\alice} \otimes m^{\alice}$ allow users to combine and split their coins in a value-preserving manner, and the $\gamma^{\alice,\bob}_n$ maps allow them to exchange coins.

Now, a ledger is a (syntactic) morphism $a : I \to A$ of $\mathcal{C}(\mathbb{N}_\nu)$. A transaction to be included in $a$ consists of a transformation $f : X \to Y$ of $\mathcal{C}(\mathbb{N}_\nu)$ along with information about which outputs of $a$ are to be the inputs of the transformation, which we package as $t = \pi_t(1 \otimes f \otimes 1) : A \to B$. The result of including transaction $t$ in ledger $a$ is then the composite ledger $t \circ a : I \to B$. Put another way, a ledger is given by a list of transformations in $\mathcal{C}(\mathbb{N}_\nu)$:
\[I \stackrel{t_1}{\longrightarrow} A_1 \stackrel{t_2}{\longrightarrow} \cdots \stackrel{t_k}{\longrightarrow} A_k \]

For the purpose of illustration, we differentiate between $m + n$ and $m \otimes n$ in our string diagrams for $\mathbb{N}_\nu$. We do so by means of the string diagrams for (not necessarily strict) monoidal categories (see e.g. \cite{Coc17}), as in:

\[\includegraphics[height=2cm]{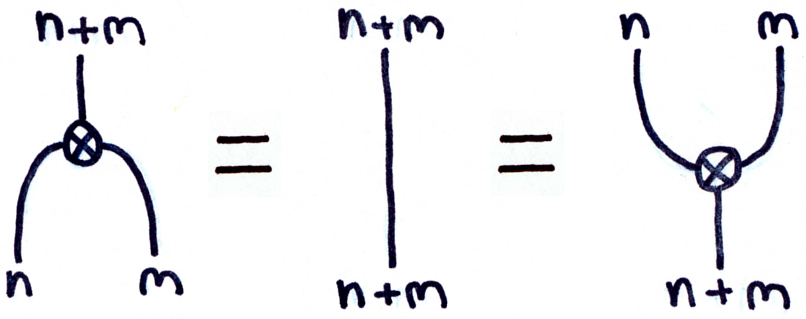}\]

\noindent Now, suppose we have a ledger $a : I \to \nu^{\carol}_7 \otimes \nu^{\alice}_5$:

\[\includegraphics[height=1.7cm]{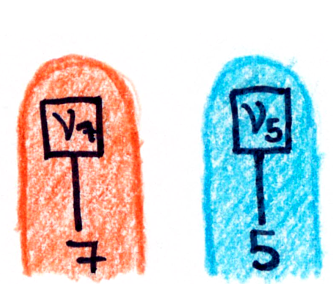}\]

\noindent and resource transformations $f_1,f_2,f_3$ defined, respectively, by:


\[\includegraphics[height=3.5cm]{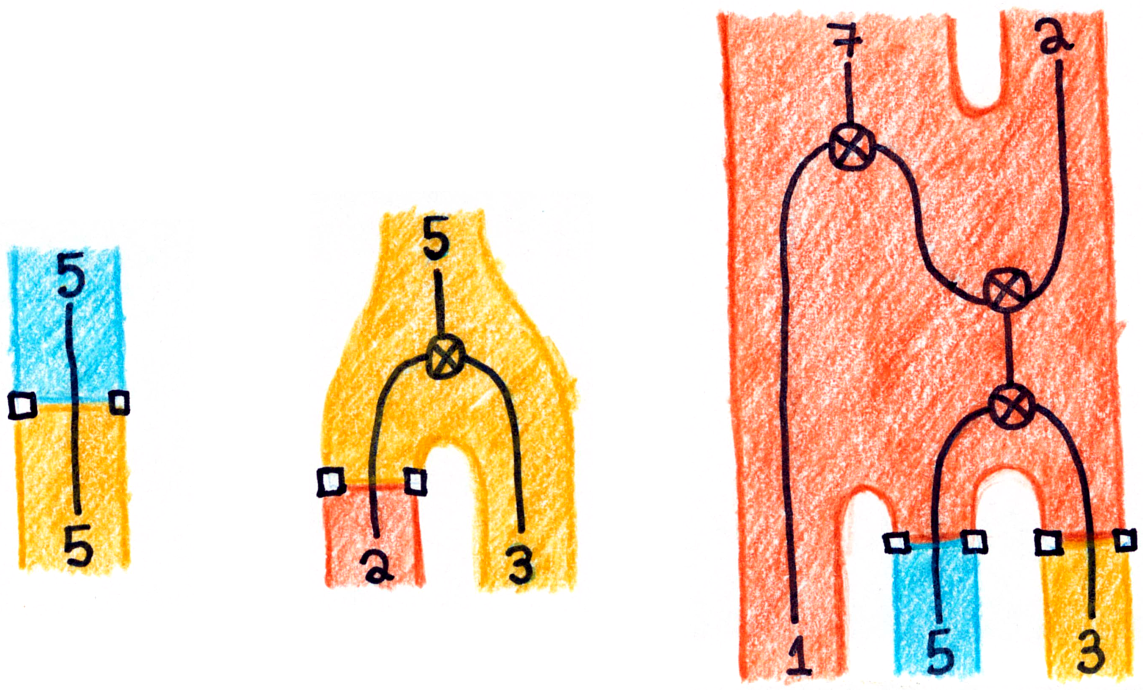}\]

\noindent Now, form transaction $t_1 = (1_7^{\carol} \otimes f_1)$ and append it to $a$ to obtain $t_1 \circ a$
\[\includegraphics[height=2.5cm]{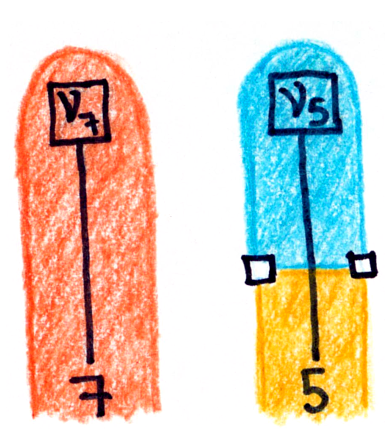}\]

\noindent Next, form transaction $t_2 = (1_7^{\carol} \otimes f_2)$ and append it to obtain $t_2 \circ t_1 \circ a$
\[\includegraphics[height=3.5cm]{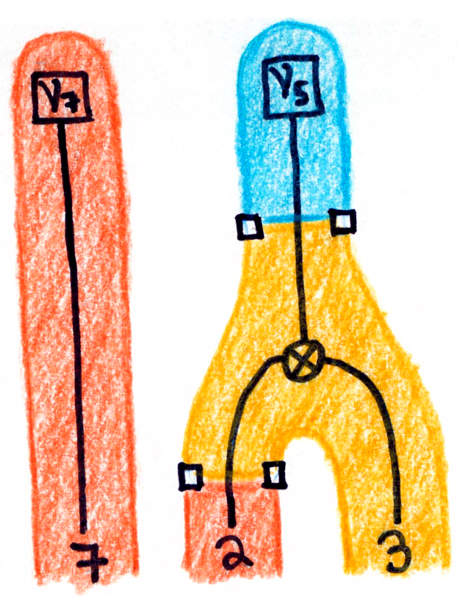}\]

\noindent Finally, form transaction $t_3 = (f_3 \otimes 1_3^{\bob})$ and append it to obtain $t_3 \circ t_2 \circ t_1 \circ a$
\[\includegraphics[height=6cm]{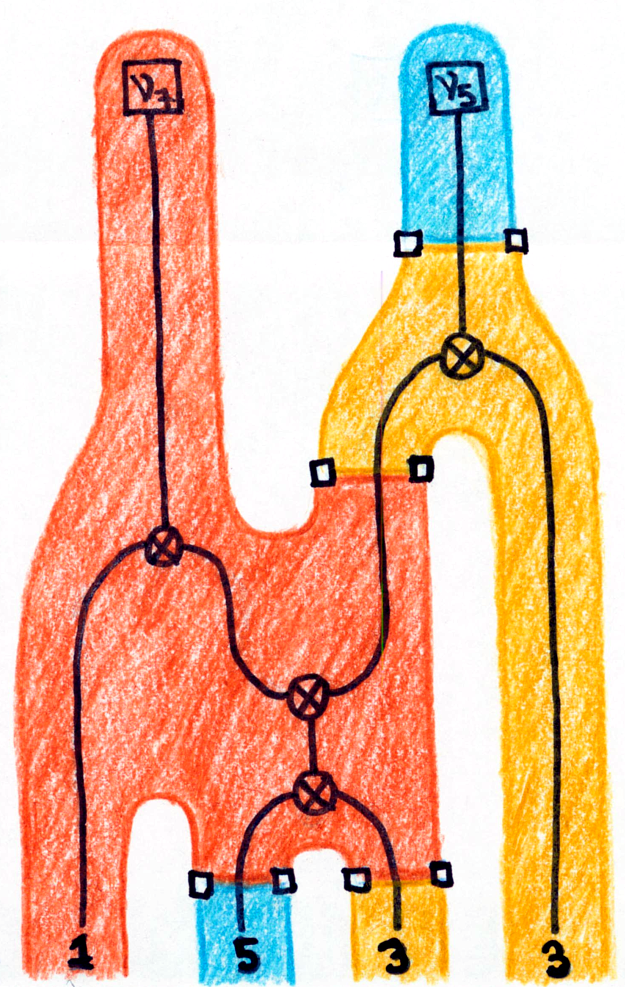} \]

In this manner, we capture the evolution of the ledger over time. Of course, we can also reason about whether two sequences of transactions result in the same ledger state by comparing the corresponding morphisms for equality, although in the case of $\mathcal{C}(\mathbb{N}_\nu)$ there isn't much point, since all morphisms $A \to B$ are necessarily equal. 

\section{Conclusions and Future Work}

We have seen how the resource theoretic interpretation of monoidal categories, and in particular their string diagrams, captures the sort of material history that concerns ledger structures for blockchain systems. Additionally, we have shown how to freely add a notion of ownership to such a resource theory, and that the resulting category is equivalent to the original one. We have also shown that these resource theories with ownership admit an intuitive graphical calculus, which is more or less that of monoidal functors and natural transformations. Finally, we have used our machinery to construct a simple ledger structure and show how it might be used in practice. 

While we do not claim to have solved the problem of providing a rigorous foundation for the development of ledger structures in its entirety, we feel that our approach shows promise. There are a few differnt directions for future research. One is the development of categorical models for more sophisticated ledger structures, with the eventual goal being to give a rigorous formal account of smart contracts. Another is to explore the connections of the current work with formal treatments of accounting, such as \cite{Kat98}.

\newpage

\bibliographystyle{plain}
\bibliography{citations}
\end{document}